\documentclass[]{article}

\usepackage{url}
\usepackage{amsmath,amssymb}
\usepackage{amsthm}
\usepackage{tikz}
\usepackage{comment}
\usepackage[english]{babel}
\usepackage[colorlinks=true,citecolor=black,linkcolor=black]{hyperref}
\usepackage[color=green!20]{todonotes}
\usepackage{xspace} 
\usepackage{scalefnt}
\usepackage{a4wide}	

\newtheorem{theorem}{Theorem}
\newtheorem{observation}{Observation}
\newtheorem{proposition}{Proposition}

\newcommand{\e}{\text{e}}
\newcommand{\vc}{\text{vc}}
\newcommand{\cu}{\text{c}}

\newcommand{\Tau}{\mathrm{T}}

\begin{document}

\title{A note on Edge Isoperimetric Numbers and Regular Graphs}


\author{\'Edouard Bonnet\footnote{Institute for Computer Science and Control, Hungarian Academy of Sciences (MTA SZTAKI). bonnet.edouard@sztaki.mta.hu} \qquad Florian Sikora \footnote{PSL, Universit\'{e} Paris-Dauphine, LAMSADE UMR CNRS 7243, France. florian.sikora@dauphine.fr}}

\date{}



%
%

\maketitle

\begin{abstract}
This note resolves an open problem asked by Bezrukov in the open problem session of IWOCA 2014.
It shows an equivalence between regular graphs and graphs for which a sequence of invariants presents some symmetric property.
We extend this result to a few other sequences.
\end{abstract}





\section{Introduction}


Let $G=(V,E)$ be a simple undirected connected graph, with $|V|=n$ and $|E|=m$. 
For any pair of integers $a \leqslant b$, we denote by $[a,b]$ the set $\{a,a+1,\ldots,b-1,b\}$ and we set $[a]:=[1,a]$.
For any set $S$ and integer $i$, ${S \choose i}$ denotes the set of all the subsets of $S$ of size $i$.
All the definitions in this section are relative to $G$ but we choose to indicate the graph in subscript notation only if it would be ambiguous otherwise.
For any subset $S \subseteq V$, we define the following values:

\begin{itemize}
	\item $\e(S) = |\{uv \in E$ $|$ $u \in S$ and $v \in S\}|$,
	\item $\vc(S) = |\{uv \in E$ $|$ $u \in S$ or $v \in S\}|$,
	\item $\cu(S) = |\{uv \in E$ $|$ $u \in S$ and $v \notin S\}|$.
\end{itemize}

We then define, for each $i \in [0,n]$, the quantities:
\begin{itemize}
	\item $\Delta(i) = \max \{\e(S) $ $|$ $S \in {V \choose i}\}$,
	\item $\Sigma(i) = \min \{\e(S) $ $|$ $S \in {V \choose i}\}$,
	\item $\Tau(i) = \max \{\vc(S) $ $|$ $S \in {V \choose i}\}$,
	\item $\Upsilon(i) = \min \{\vc(S) $ $|$ $S \in {V \choose i}\}$,
	\item $\Theta(i) = \max \{\cu(S) $ $|$ $S \in {V \choose i}\}$,
	\item $\Phi(i) = \min \{\cu(S) $ $|$ $S \in {V \choose i}\}$.
\end{itemize}

By definition, observe that $\Delta(i)$ corresponds to the \textsc{$i$-Densest Subgraph} problem (the problem of finding $i$ vertices of the graph s.t. the number of edges induced by these vertices is maximum), $\Sigma(i)$ to \textsc{$i$-Sparsest Subgraph} (the problem of finding $i$ vertices of the graph s.t. the number of edges induced by these vertices is minimum), $\Tau(i)$ to \textsc{Max $i$-Vertex Cover} (also known as \textsc{Partial Vertex Cover} or \textsc{Max $i$-Coverage}, the problem of finding a set of $i$ vertices of the graph s.t. the number of edges with at least one endpoint in this set is maximum), $\Upsilon(i)$ to \textsc{Min $i$-Vertex Cover} (the problem of finding a set of $i$ vertices of the graph s.t. the number of edges with at least one endpoint in this set is minimum), $\Theta(i)$ to \textsc{Max $(i,n-i)$-Cut} (the problem of finding a set of $i$ vertices of the graph s.t. the number of edges with one endpoint in this set and the other endpoint in the complement is maximum) and $\Phi(i)$ to \textsc{Min $(i,n-i)$-Cut} (the problem of finding a set of $i$ vertices of the graph s.t. the number of edges with one endpoint in this set and the other endpoint in the complement is minimum).

In combinatorial optimization, these problems are known as \emph{cardinality-constrained}, or {size-constrained}, or even \emph{fixed-cardinality} graph problems.
They all inherit NP-hardness from their non cardinality-constrained variant, and have been extensively studied mainly from the standpoint of parameterized complexity (see for instance \cite{Cai08}) and approximation (see for instance \cite{Bruglieri2006}). 

In graph theory, those problems are mostly known as \emph{edge isoperimetric graph problems}~\cite{Bezrukov1999,Bezrukov2003} analogously to the \emph{isoperimetric problem} of the ancient Greeks, where one has to find a closed curve with a fixed given perimeter in order to maximize the area of the interior.

Value $\Phi(i)$ is sometimes called the \emph{edge isoperimetric parameter}.
\emph{Edge isoperimetric inequality} consists of upper bounding $\frac{\Phi(i)}{i}$.
For instance, it is known that $\frac{\Phi_{Q_d}(i)}{i} \geqslant d - \log i$ where $Q_d$ is the hypercube graph of dimension~$d$~\cite{Hoory06,Bollobas1986}.
Such inequalities lower bound the edge expansion of graphs and therefore indicate how well they perform as expanders.


Finally, we define, for each $i \in [n]$, the consecutive differences:
\begin{itemize}
	\item $\delta(i) = \Delta(i) - \Delta(i-1)$,
	\item $\sigma(i) = \Sigma(i) - \Sigma(i-1)$,
	\item $\tau(i) = \Tau(i) - \Tau(i-1)$,
	\item $\upsilon(i) = \Upsilon(i) - \Upsilon(i-1)$,
	\item $\theta(i) = \Theta(i) - \Theta(i-1)$,
	\item $\phi(i) = \Phi(i) - \Phi(i-1)$,
\end{itemize}

For any $s \in \{\delta,\sigma,\tau,\upsilon,\theta,\phi\}$, we say that the sequence $s(1), \ldots, s(n)$ (called \emph{$s$-sequence}) is \emph{symmetric} if the value $s(i) + s(n-i+1)$ does not depend on $i$.
In other words, the $s$-sequence is said to be symmetric if, for each $i \in [n]$, $s(i)+s(n-i+1)=s(1)+s(n)$.
It was conjectured by Bezrukov (open problem session of IWOCA 2014~\cite{openiwoca}) that the $\delta$-sequence of a graph $G$ is symmetric iff $G$ is regular. 
We prove this conjecture and show that it extends to other sequences:

\begin{theorem}\label{thm:main}
For any $s \in \{\delta,\sigma,\tau,\upsilon\}$ and graph $G$: 

$G$ is regular if and only if its $s$-sequence is symmetric.
\end{theorem}

We also observe that the $\theta$-sequence and the $\phi$-sequence are always symmetric.

\section{Symmetry of sequences and regular graphs}

\subsection{Max and Min $(k,n-k)$-Cut}


\begin{observation}
Every graph has a symmetric $\theta$-sequence (resp.~$\phi$-sequence). 
\end{observation}

\begin{proof}
By definition of a $(k,n-k)$-cut, $\Theta(i) = \Theta(n-i)$, $\forall i \in [0,n]$. 
Thus, $\theta(i) = \Theta(i) - \Theta(i-1) = \Theta(n-i) - \Theta(n-i+1) = -\theta(n-i+1)$.
So, $\theta(i)+\theta(n-i+1)=0$,~$\forall~i~\in~[n]$.
%
%
The same argument carries over to the $\phi$-sequence.
\end{proof}

\subsection{$k$-Densest and $k$-Sparsest}

We prove the conjecture (given in~\cite{openiwoca}) for the $\delta$-sequence, namely:

\begin{proposition}\label{thm:delta}
$G$ is regular iff its $\delta$-sequence is symmetric.
\end{proposition}

\begin{proof}
Let $n$ be the number of vertices of $G=(V,E)$, $m$ its number of edges, and $d$ its minimum degree.
Obviously, $\Delta(n) = m$ and $\Delta(n-1) = m-d$.
Indeed, an $(n-1)$-densest contains all the vertices of $V$ except one vertex of minimum degree.
Thus, $\delta(n)=m-(m-d)=d$.
Also, $\delta(1)=0$.
Therefore, the $\delta$-sequence is symmetric iff for each $i \in [n]$, $\delta(i)+\delta(n-i+1)=d$.

$\Rightarrow$ : We assume that $G$ is $d$-regular. 
We  observe that $2\Delta(i) + \Phi(i) = id$, $\forall i \in [n]$. 
So, $\delta(i) = \Delta(i) - \Delta(i-1) $ $= \frac{1}{2}(id - \Phi(i)) - \frac{1}{2}((i-1)d - \Phi(i-1)) = \frac{1}{2}(d - \Phi(i) + \Phi(i-1))$. 
Similarly, $\delta(n-i+1) = \frac{1}{2}(d - \Phi(n-i+1) + \Phi(n-i))$. 
As $\Phi(n-i+1)=\Phi(i-1)$ and $\Phi(n-i)=\Phi(i)$, it holds that $\delta(i)+\delta(n-i)=\frac{1}{2}(d-\Phi(i)+\Phi(i-1)+d-\Phi(i-1)+\Phi(i))=d$.

$\Leftarrow$ : We now assume that the $\delta$-sequence is symmetric, that is $\delta(i) + \delta(n-i+1) = d$, $\forall i \in [n]$. 
We show by induction on $i$ that $\Delta(n-i) = m - id + \Delta(i)$, $\forall i \in [0,n]$. 
This is true for $i=0$ since $\Delta(n) = m$ and $\Delta(0)=0$.
We suppose the property true for $j-1$ for some $j \in [n]$. 
$\Delta(n-j) = \Delta(n-(j-1)) - \delta(n-j+1)$ by definition of $\delta$. 
By induction hypothesis, one gets $\Delta(n-j) = m - (j-1)d + \Delta(j-1) - \delta(n-j+1)$. 
Since the $\delta$-sequence is symmetric, $\Delta(n-j) = m - (j-1)d + \Delta(j-1) + \delta(j) - d$. 
Therefore, $\Delta(n-j) = m - jd + \Delta(j)$, ending the induction.

The property for $i=n$ yields $\Delta(0)=m-nd+\Delta(n)$.
Hence, $m=\frac{nd}{2}$.
As $m=\frac{1}{2}\Sigma_{v \in V}\text{deg}(v)$, the previous identity only holds if $\forall v \in V$, $\text{deg}(v)=d$, that is graph $G$ is $d$-regular.

\end{proof}

\begin{proposition}
$G$ is regular iff its $\sigma$-sequence is symmetric.
\end{proposition}
\begin{proof}
We prove that the $\delta$-sequence is symmetric iff the $\sigma$-sequence is symmetric and conclude with Proposition~\ref{thm:delta}.
For that we show that $\delta_G(i) + \delta_G(n-i+1) + \sigma_{\overline G}(i) + \sigma_{\overline G}(n-i+1) = n-1$, $\forall i \in [n]$, where $\overline G$ is the complement graph of $G$. 
Thus, $\delta_G(i) + \delta_G(n-i+1)$ does not depend on $i$ iff $\sigma_{\overline G}(i) + \sigma_{\overline G}(n-i+1)$ does not depend on $i$.

Observe that $\Delta_G(i) = {i \choose 2} - \Sigma_{\overline G}(i)$.
Thus, $\delta_G(i) + \delta_G(n-i+1) = \Delta_G(i) - \Delta_G(i-1) + \Delta_G(n-i+1) - \Delta_G(n-i) =  {i \choose 2} - \Sigma_{\overline G}(i) -  ({i-1 \choose 2} - \Sigma_{\overline G}(i-1)) + {n-i+1 \choose 2} - \Sigma_{\overline G}(n-i+1) - ({n-i \choose 2} - \Sigma_{\overline G}(n-i)) = - (\Sigma_{\overline G}(i) - \Sigma_{\overline G}(i-1)) + i - 1 - (\Sigma_{\overline G}(n-i+1) - \Sigma_{\overline G}(n-i)) + n - i = - \sigma_{\overline G}(i) - \sigma_{\overline G}(n-i+1)  + n - 1$. 
This ends the proof since the $\delta$-sequence is symmetric iff $G$ is regular iff $\overline G$ is regular.
\end{proof}

\subsection{Max and Min $k$-Vertex Cover}

\begin{proposition}
$G$ is regular iff its $\tau$-sequence is symmetric.
\end{proposition}

\begin{proof}
Equivalently, we show that the $\tau$-sequence is symmetric iff the $\sigma$-sequence is symmetric, and more precisely that $\tau(i) + \tau(n-i+1) =\sigma(i) + \sigma(n-i+1)$.
For any $i \in [n]$, the complement of a $(n-i)$-sparsest is a maximum $i$-vertex cover (that is, a set of $i$ vertices touching the largest number of edges) and $\Tau(i) + \Sigma(n-i) = m$.
So, $\tau(i) + \tau(n-i+1) = \Tau(i) - \Tau(i-1) + \Tau(n-i+1) - \Tau(n-i)$
$=- \Sigma(n-i) + \Sigma(n-i+1) - \Sigma(i-1) + \Sigma(i) =\sigma(i) + \sigma(n-i+1)$.
\end{proof}

\begin{proposition}
$G$ is regular iff its $\upsilon$-sequence is symmetric.
\end{proposition}

\begin{proof}
For any $i \in [n]$, the complement of a $(n-i)$-densest is a minimum $i$-vertex cover (that is, a set of $i$ vertices touching the smallest number of edges) and $\Upsilon(i) + \Delta(n-i) = m$.
Then, similarly to the previous proof, we can show that the $\upsilon$-sequence is symmetric iff the $\delta$-sequence is symmetric.
\end{proof}

\bibliographystyle{abbrv} 
\bibliography{eip}

\end{document}